\documentclass{article}

\usepackage{times}
\usepackage{authblk}
\usepackage{todonotes}
\usepackage[utf8]{inputenc}
\usepackage{caption}
\usepackage[a4paper, total={6.5in, 10in}]{geometry}
\usepackage{setspace}
\usepackage{multirow}
\usepackage{changepage}
\usepackage{array}
\usepackage{amsmath}
\usepackage{amsthm}
\usepackage{amssymb} 
\usepackage{algorithm}
\usepackage{mathtools}
\usepackage{hyperref}
\usepackage[export]{adjustbox}
\usepackage{enumitem}
\usepackage{booktabs}

\usepackage[noend]{algpseudocode}
\algnewcommand\Input{\item[\textbf{Input:}]}
\algnewcommand\Output{\item[\textbf{Output:}]}

\newtheorem{definition}{Definition}
\newtheorem{theorem}{Theorem}

\def \ie {\textit{i.e.}}

\DeclareMathOperator*{\argmin}{arg\,min}
\DeclareMathOperator*{\argmax}{arg\,max}
\DeclarePairedDelimiter{\ceil}{\lceil}{\rceil}

\newcommand{\R}{\mathbb{R}}
\newcommand{\N}{\mathbb{N}}

\newcommand{\RR}{Q^*}

\title{Traffic networks are vulnerable to disinformation attacks}

\author[a]{Marcin Waniek}
\author[b]{Gururaghav Raman}
\author[a]{Bedoor AlShebli}
\author[b,*]{\\Jimmy Chih-Hsien Peng}
\author[a,*]{Talal Rahwan}

\affil[a]{Computer Science, New York University, Abu Dhabi, United Arab Emirates.}
\affil[b]{Department of Electrical and Computer Engineering, National University of Singapore, Singapore.}
\affil[*]{\footnotesize Joint corresponding authors. E-mail: jpeng@nus.edu.sg; talal.rahwan@nyu.edu}

\date{}

\begin{document}
\maketitle

\begin{abstract}
Disinformation continues to attract attention due to its increasing threat to society. Nevertheless, a disinformation-based attack on critical infrastructure has never been studied to date. Here, we consider traffic networks and focus on fake information that manipulates drivers' decisions to create congestion. We study the optimization problem faced by the adversary when choosing which streets to target to maximize disruption. We prove that finding an optimal solution is computationally intractable, implying that the adversary has no choice but to settle for suboptimal heuristics. We analyze one such heuristic, and compare the cases when targets are spread across the city of Chicago vs.~concentrated in its business district. Surprisingly, the latter results in more far-reaching disruption, with its impact felt as far as 2 kilometers from the closest target. Our findings demonstrate that vulnerabilities in critical infrastructure may arise not only from hardware and software, but also from behavioral manipulation.
\end{abstract}

\section*{Introduction}
The ubiquity of social media and the internet has made it easier than ever to spread disinformation~\cite{vosoughi2018spread,lazer2018science,fletcher2018measuring, pennycook2019fighting}.
Exacerbating this phenomenon are the recent advances in machine learning and the rise of social bots, allowing (dis)information to be delivered to a target audience at an unprecedented scale~\cite{nyt2018cambridgeanalytica, shao2018spread, vosoughi2018spread}.
Disinformation has not only grown in reach but also in sophistication, with its manifestations ranging from counterfactual social media posts and manipulated news stories, to deep fake videos~\cite{scheufele2019science, bovet2019influence, grinberg2019fake, del2016spreading, bennett2018disinformation, ferrara2017disinformation}. It is, therefore, unsurprising that disinformation has come to be considered a serious threat to society~\cite{wef2013,del2016spreading}.

Several disinformation campaigns have recently garnered significant attention from the public and the scientific community alike.
Such campaigns have been used~\cite{lazer2018science,scheufele2019science,pennycook2019fighting} to shape narratives in the public debate on various issues including healthcare, vaccination, and climate change, to name just a few. Some have even argued that disinformation is being \emph{weaponized} to manipulate the long-term decisions of a society~\cite{singer2018likewar,SenateTestimony}, e.g., during the 2016 US presidential elections and the Brexit campaign, when strategically-created propaganda, conspiracy theories, and fake social media posts were used to manipulate undecided voters~\cite{nyt2018cambridgeanalytica, grinberg2019fake, bovet2019influence}.

Nevertheless, the possibility that a malicious actor could use disinformation in a targeted attack to influence social behavior within a limited time has not been considered to date. 
If such an attack is plausible and indeed effective, it underscores an important but largely-overlooked vulnerability in complex systems whose behavior emerges from the collective decisions of the individuals therein. This is particularly alarming when the system under consideration is critical infrastructure, especially since the World Economic Forum has identified cyberattacks on critical infrastructure as one of the global risks faced in 2020~\cite{wef2020}. In this context, demonstrating the effectiveness of a disinformation-based attack may have important policy implications in securing critical infrastructure, taking into consideration the possibility of manipulating social behavior.

Here, we focus on the influence of disinformation on urban road networks, where individual drivers constantly make decisions about their routes, timing, and destination, all of which collectively shape the city-wide traffic.
The behavior of drivers has already been shown to contribute significantly towards creating bottlenecks in the traffic flow~\cite{li2015percolation,sugiyama2008traffic,ccolak2016understanding}. Thus, if an adversary is able to manipulate driving behaviors, they can potentially create massive disruptions and spill-over effects similar to those observed during the 2013 Fort Lee lane closure in New Jersey~\cite{BridgeGate2013}.
Nevertheless, despite the abundance of research on the resilience of traffic networks to external attacks~\cite{national2017cybersecurity, chen2018exposing, levy2015cyber, vivek2019cyberphysical, laszka2016vulnerability, feng2018vulnerability, huq2017cyberattacks, eykholt2018robust}, current studies focus predominantly on vulnerabilities introduced by new technologies such as intelligent traffic signaling, internet-connected vehicles, and self-driving vehicles. In contrast, security vulnerabilities arising from disinformation-based attacks have not been considered to date. 
Against this background, we ask the following questions: Can an adversary manipulate the behavior of drivers at a city scale? If so, what would be the extent of the disruption that such an attack can cause?

\section*{Results}

We consider a scenario where an adversary spreads false information, e.g., a purported accident or heavy congestion at a particular location, with the aim of manipulating the routes taken by drivers in a city. One approach could be as simple as walking along a road while carrying a number of mobile phones with Google Maps turned on. As demonstrated very recently~\cite{wapo2020googlemaps}, this effectively tricks Google's algorithms by giving the illusion that multiple vehicles are moving very slowly along that road, which in turn causes Google Maps to indicate heavy traffic and reroute all its users accordingly. Alternatively, a more sophisticated adversary may broadcast traffic alerts notifying people e.g., via SMS or social media, of a fake accident or congestion at a particular location. Regardless of how such disinformation is being spread, a proportion of the recipients may follow-through on its advice and reroute their trips to avoid the specified location, while the remaining recipients simply ignore the notification and maintain their original routes.

\begin{figure}[t]
	\centering
	\includegraphics[width=0.65\columnwidth]{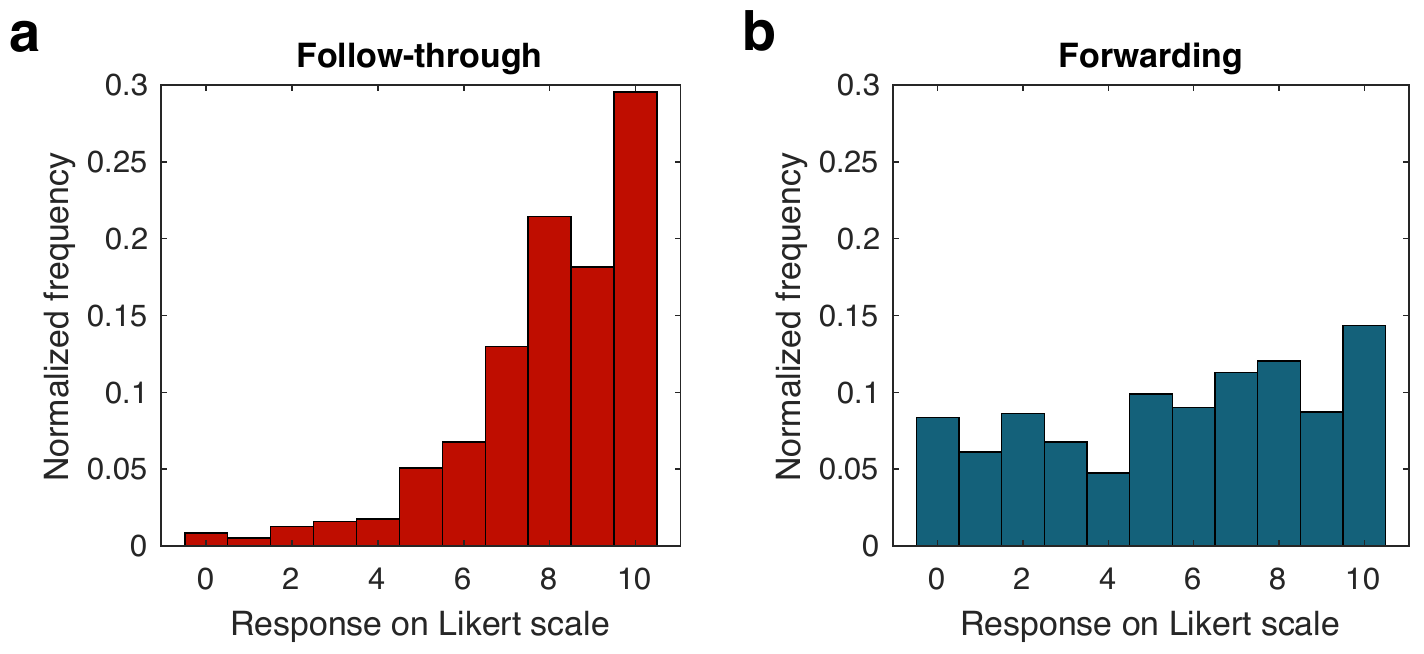}
	\caption{\textbf{Assessing propensities to follow-through and forward disinformation notifications.} \textbf{a},~Normalized histogram of participants' responses on a Likert scale from 0 to 10 indicating their propensity to follow-through on the notification. \textbf{b},~Same as (\textbf{a}) but for the propensity to forward the notification to friends. For (\textbf{a}) and (\textbf{b}), the number of data points $n=$ 3,301.}
	\label{Infographic_Traffic}
\end{figure}

To analyze this scenario, we begin by assessing how many of the recipients of such a notification will actually follow-through.
To this end, we ran a survey on Mechanical Turk ($n=$ 3,301), where each participant was shown an SMS notification whose sender is \emph{`SMSAlert'} and whose content is: \emph{``Accident on `X' Road. Please use alternative routes. Be safe!''}. The participants were then asked to specify on a Likert scale from 0 to 10 their likelihood to follow-through on the notification.
Furthermore, since the reach of the disinformation may increase if the recipients forward its content to others, e.g., via social media, the participants in our survey were also asked to specify their likelihood to forward the notification to their friends. Although the participants' responses in reality may not be exactly what is reported in the survey, these responses nevertheless give us an insight into how receptive people are towards such traffic alerts.
The survey questionnaire is presented in Supplementary Note~1, and the results are depicted in Figure~\ref{Infographic_Traffic}.
Worryingly, the histogram of the follow-through responses is skewed to the left, with 89\% of the participants reporting their propensity to follow-through to be 6 or above. These results suggest that such notifications would alter the behavior of the majority of recipients. On the other hand, we find a more even spread in the forwarding responses, with 55\% of the respondents reporting a propensity of 6 and above, implying that people are more likely to follow-through than to forward the notification. Nevertheless, these results suggest that the overall reach of the attack is likely to increase beyond the initial set of recipients due to the forwarding of the disinformation.

Naturally, the impact of such a behavioral manipulation attack not only depends on the reach of the disinformation, but also on its content. In our scenario, the impact of the attack on the traffic flows not only depends on how many people deviate away from the target, but also on where the target is located. For example, an attack targeting a lightly-used road in the suburbs is unlikely to cause a major disruption, especially if there are several alternative roads that the drivers could take to avoid the target. Furthermore, the adversary may decide to attack multiple targets at the same time, in which case there might be a compound impact on traffic, especially if all targets are concentrated in a single neighborhood rather than dispersed across the city.

Essentially, the adversary must solve an optimization problem which involves choosing the targets that maximize the traffic disruption; we call this the \textit{problem of Maximizing Disruption}. We analyze the computational complexity of this problem to understand the theoretical limits of the capability of such an adversary. To this end, we represent the road network as a directed graph, $G=(V,E)$, where $V$ is the set of nodes, each representing an intersection or a road end, and $E$ is the set of directed edges, each representing a single direction of traffic movement, i.e., a one-way street is represented by a single edge while a two-way street is represented by two separate edges. We then define $R=\{r_1,\ldots,r_{|R|}\}$ as the set of vehicle rides, each specifying the start node, the end node, and the time of day when the ride starts, but not specifying the route itself. Now, let $b$ be the ``budget'', i.e., the maximum number of targets that the adversary can afford to attack. Finally, given a traffic model, $M$, let $f$ be the objective function, which encapsulates the quality of traffic into a single number. Then, the problem of Maximizing Disruption involves removing at most $b$ streets from the graph, $G$, in order to minimize the traffic quality, $f$, given the rides in $R$ whose routes are determined based on the traffic model $M$; see Definition~S1 in Supplementary Note~4.

In our complexity analysis, we focus on one particular objective function, $f^*$, that reflects the average time taken to complete each ride in $R$; see Definition~\ref{def:avg-time-evaluation-function} in Methods. 
We also focus on a cellular-automata traffic model, $M^*$, which is used to simulate the rides in the city; see Methods.
Our theoretical analysis shows that given the objective function $f^*$ and the model of traffic $M^*$, the problem of Maximizing Disruption is at least as hard as any NP (non-deterministic polynomial time) problem, implying that no known algorithm can solve it in polynomial time; see Theorem~\ref{thrm:npc-our-model}, the proof of which is detailed in Supplementary Note~5.

\begin{theorem}
\label{thrm:npc-our-model}
The problem of Maximizing Disruption is NP-complete given the objective function $f^*$ and the traffic model $M^*$.
\end{theorem}

To determine whether this computational intractability arises from the complexity of the traffic model $M^*$, we consider a simpler alternative, $M^\varnothing$, which completely disregards the interdependencies between the rides; see Definition~\ref{def:simple-model} in Methods. Surprisingly, even for this simple model, the computational complexity persists; see Theorem~\ref{thrm:npc-average-distance}.

\begin{theorem}
\label{thrm:npc-average-distance}
The problem of Maximizing Disruption is NP-complete given the objective function $f^*$ and the traffic model $M^\varnothing$.
\end{theorem}

So far, we analyzed the problem of identifying an optimal set of $b$ targets that will maximize traffic disruption. 
While this problem is the primary focus of our study, we now analyze the computational complexity of a slightly different problem, which involves identifying the minimum number of targets needed to achieve a predetermined level of disruption; we call this the \textit{problem of Minimizing Targets}, see Definition~S4 in Supplementary Note~4.
Our analysis shows that this problem is APX-hard given both our model of traffic $M^*$ and the simple model $M^\varnothing$.
Intuitively, it is computationally intractable to approximate, let alone find, the optimal solution. Moreover, the problem does not admit an algorithm (a polynomial-time approximation scheme) that takes a parameter $\epsilon>1$ and, in polynomial time, finds a solution that is within a factor $\epsilon$ of being optimal. Our computational complexity results are summarized in Table~\ref{tab:theoretical-results}.

\begin{theorem}
\label{thrm:apx-hard}
The problem of Minimizing Targets is APX-hard given the objective function $f^*$ and either traffic model $M^*$ or $M^\varnothing$. In particular, the problem does not admit a polynomial-time approximation scheme (PTAS) unless P=NP.
\end{theorem}

\begin{table}[th]
    \centering
    \begin{tabular}{c c c}
    \toprule
    \textbf{Model of traffic} & \textbf{Maximizing Disruption} & \textbf{Minimizing Targets} \\
    \midrule
    $M^*$ & NP-complete & APX-hard \\
    $M^\varnothing$ & NP-complete & APX-hard \\
    \bottomrule
    \end{tabular}
    \caption{Summary of our computational complexity results.}
    \label{tab:theoretical-results}
\end{table}

Given our complexity analysis, we can safely assume that the adversary will not be able to identify an optimal set of targets in the network, and has no choice but to settle for a suboptimal solution.
With this in mind, we analyze a \textit{greedy} heuristic where the adversary starts by computing the set $\Pi$ of all shortest paths between any two nodes in the graph $G$. The heuristic then proceeds in iterations, each targeting an edge that appears in the largest number of paths, taking into consideration only the paths in $\Pi$ that were not already affected by previous iterations. The fact that this greedy heuristic selects one target at a time, rather than selecting them all at once, allows the solution to be computed much more efficiently, although this may come at the expense of solution quality.

\begin{figure}[tp]
    \centering
    \includegraphics[width=\linewidth]{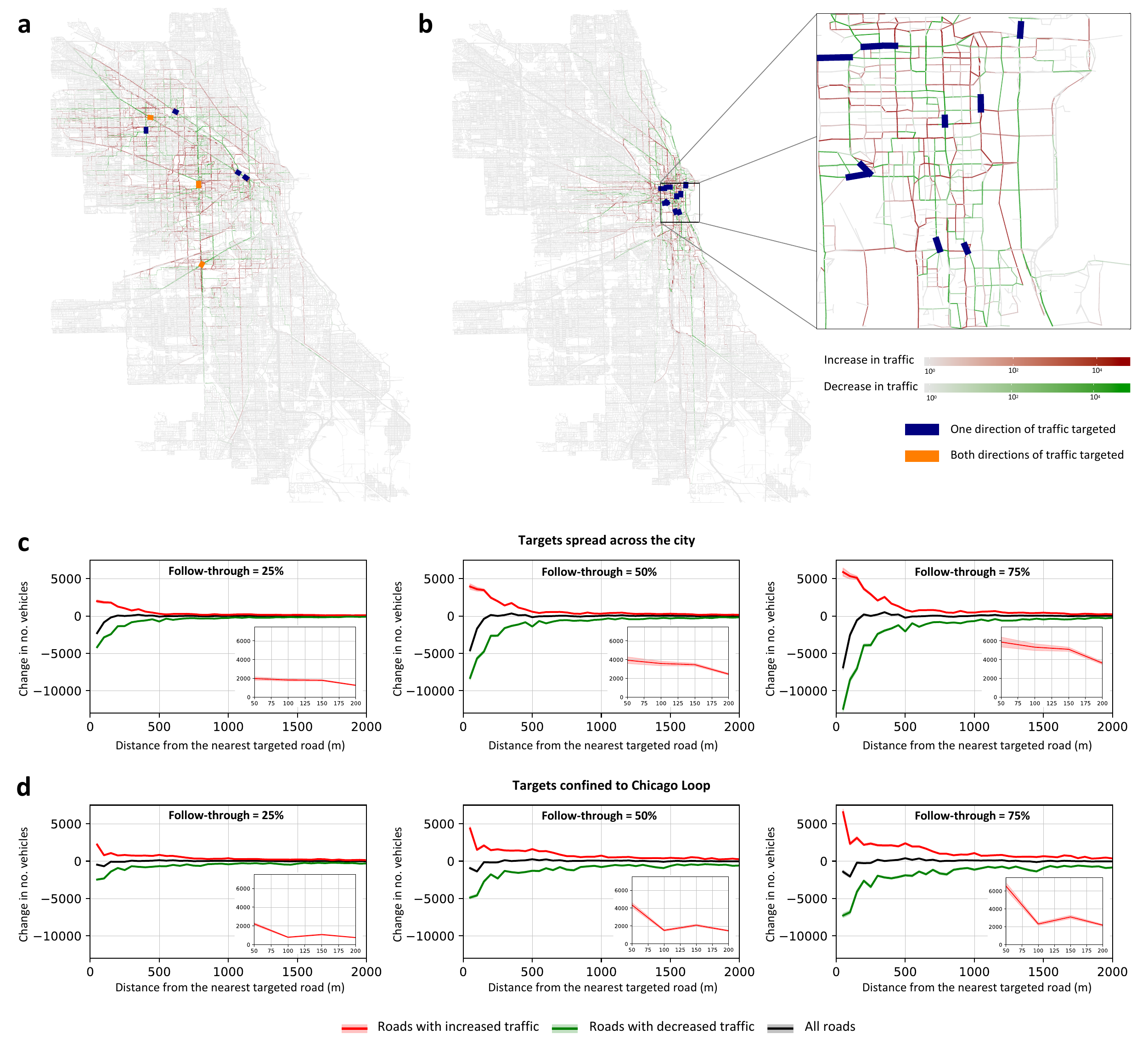}
    \caption{\textbf{Impact on streets across the city given 10 targets.} \textbf{a},~Thick lines indicate the targets selected by our greedy heuristic, and the colors of those lines indicate the number of targets in each such location; dark blue indicates locations where only a single direction of traffic was attacked, whereas orange indicates locations where both directions were attacked, with each direction counting as a separate target. For all remaining streets, the color indicates the change in the number of vehicles passing through given a follow-though rate of 50\%, with green indicating fewer vehicles, red more vehicles, and grey no change. \textbf{b},~The same as (\textbf{a}) but the greedy heuristic is confined to a particular neighborhood in downtown Chicago, namely, the Loop. \textbf{c},~Given the targets shown in (\textbf{a}) and different disinformation follow-through rates, the figure depicts the change in traffic per day as a function of the distance (in meters) from the nearest target. The black plot corresponds to all roads, whereas the red and green plots correspond to only those roads where the traffic increased and decreased, respectively. The inset figures present a zoomed-in view of the red plots in the near-vicinity of the targets. The results in (\textbf{c}) are shown as an average over 100 simulations, with the shaded ares representing the 95\% confidence intervals. \textbf{d},~The same as (\textbf{c}) but for the target locations shown in (\textbf{b}).}
    \label{fig:traffic-map}
\end{figure}

To evaluate this heuristic, we consider a scenario where the adversary targets 10 locations in the city of Chicago. We simulate one day of traffic in the city, with and without the attack, and report the observed differences in traffic throughout the day, assuming that the attack happens at the beginning of the day and lasts for 24 hours; see Methods for more details. Figure~\ref{fig:traffic-map}a shows the targets chosen by the greedy heuristic. As can be seen, these targets are spread across the city, and not grouped in any particular neighborhood. Moreover, recall that our setting considers each direction of traffic in a street to be a separate potential target; we see from the figure that the heuristic chose to target both directions of traffic in certain locations.
Given these targets, and assuming that 50\% of the drivers follow-through on the disinformation by rerouting their path to avoid the targets, Figure~\ref{fig:traffic-map}a also depicts the change in the traffic intensity as a result of the attack, i.e., the change in the number of vehicles that traverse through each street. It can be seen that the attack diverts traffic away from the targets into neighboring streets, thereby resulting in increased traffic in certain streets (shown in red), and decreased traffic in others (shown in green). Interestingly, the impact of the attack is not confined to small neighborhoods around the targets, but rather propagates across the city.

We then investigate the possibility that instead of targeting the entire city, the adversary might concentrate the attack on one critical neighborhood, e.g., its business district, in order to cause the maximum disruption in this area. To simulate this scenario, we constrained our greedy heuristic to only target locations in the Chicago Loop, which is the central business district of the city. Figure~\ref{fig:traffic-map}b highlights the locations chosen by the modified heuristic along with the resultant impact on traffic. As can be seen, when compared to Figure~\ref{fig:traffic-map}a, the impact of this concentrated attack is largely centered around the targeted neighborhood. While one may expect that most streets in this neighborhood will experience an increase in traffic, the zoomed-in part of Figure~\ref{fig:traffic-map}b shows that this is not the case, as evident by the many green streets therein. This demonstrates that the heuristic can easily be modified to control the region, but not the streets, where the traffic disruption will be concentrated. 

To further understand how the impact propagates from the targeted locations, we plot the change in the number of vehicles in different streets as a function of the distance from the nearest target, while varying the disinformation follow-through rate. Since our traffic simulation is non-deterministic, we took an average over 100 simulations; the results are depicted in Figure~\ref{fig:traffic-map}c and \ref{fig:traffic-map}d for the two cases when the targets are spread across the city, and when they are confined to the Chicago Loop, respectively. The overall trend (depicted in black) is further disaggregated into the streets that experience lighter traffic (green) and those that experience heavier traffic (red). The insets zoom in on the latter category of streets within 200 meters from a target, to facilitate comparison across the different settings. 
We observe that the closer a street is to a target, the greater is the impact on traffic, regardless of whether it is an increase (red plot) or a decrease (green plot) in the number of vehicles. 
The figures also show that higher follow-through rates result in greater impact on traffic. 
Comparing the insets of Figure~\ref{fig:traffic-map}c and \ref{fig:traffic-map}d reveals that when the targets are concentrated in one neighborhood the congestion is higher within 50 meters from the targets, e.g., given a 75\% follow-through rate, the number of vehicles throughout the day increases by about 6,500 when the targets are concentrated, as opposed to approximately 5,900 when they are not. 
However, as the distance from the nearest target increases, the impact on traffic fades away at a greater pace when targets are concentrated, e.g., at a distance of 100 meters from the targets and given a 75\% follow-through rate, the number of vehicles throughout the day increases by about 2,300 when the targets are concentrated, as opposed to approximately 5,300 when they are not.
Finally, let us comment on the reach of the disruption caused by the attack. Figure~\ref{fig:traffic-map}c and \ref{fig:traffic-map}d show that as the follow-through rate increases, the disruption extends farther away from the targets. The figures also show that the reach is greater when the targets are concentrated, with a considerable impact being felt as far as 2 kilometers away from the closest target in the case of 75\% follow-through rate. This could be due to the synergy between the different targets when they are within close proximity to one another.

\begin{figure}[t]
    \centering
    \includegraphics[width=0.6\columnwidth]{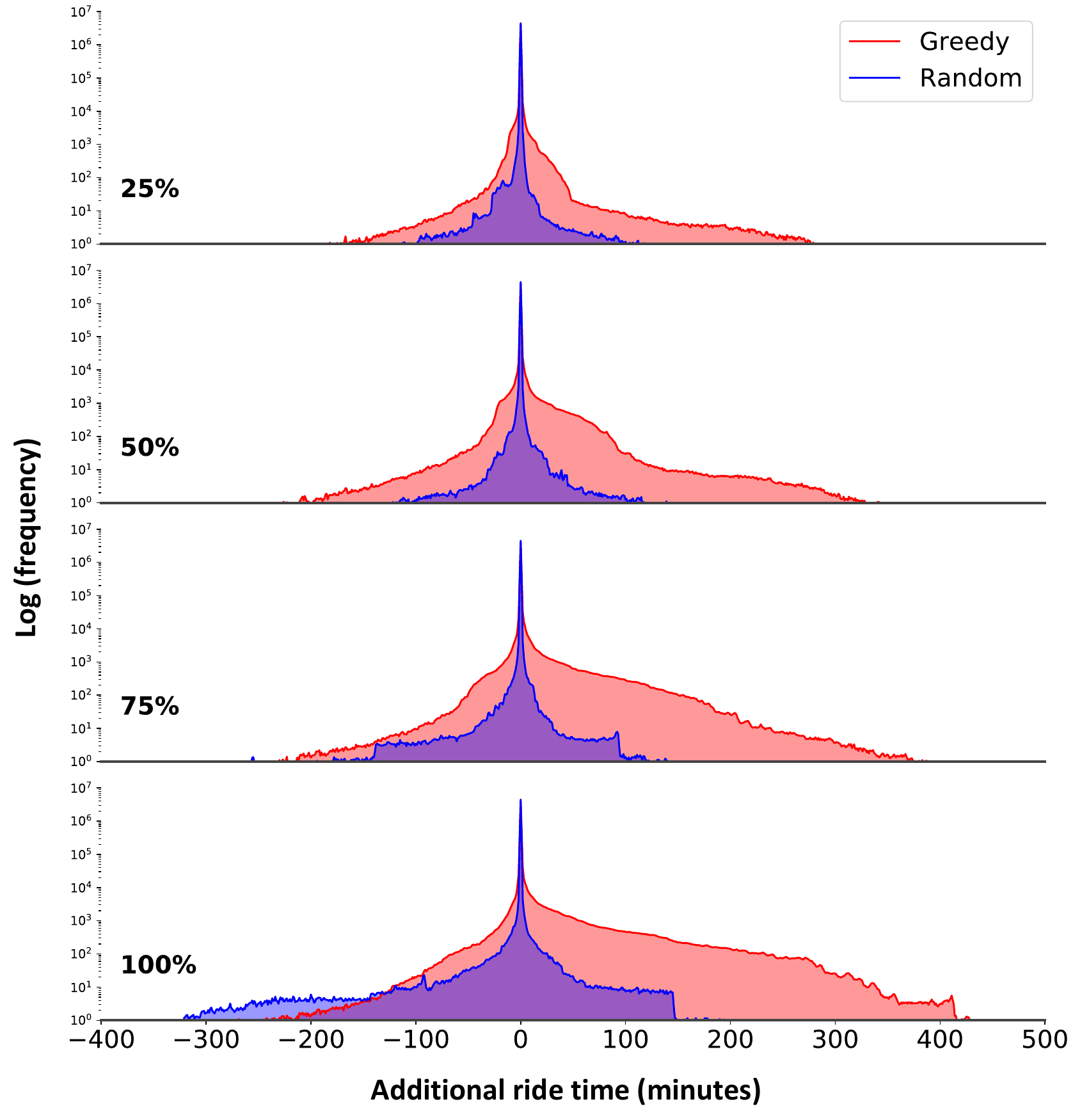}
    \caption{\textbf{Impact on rides across the city.}~Given varying disinformation follow-through rates (indicated in bold), the distribution of the additional ride time experienced by drivers after an attack on 10 targets across the city, chosen either by our greedy heuristic or at random. All results are shown as an average over 100 simulations.}
    \label{fig:delay-detour}
\end{figure}

So far, we have studied the impact of the disinformation attack on the streets across the city, showing that some streets experience increased traffic, while others experience lighter traffic. However, we still do not know the distribution of the delays experienced by the individual drivers. To this end, Fig.~\ref{fig:delay-detour} presents the distributions of the additional ride time given 10 targets chosen by our greedy heuristic while varying the follow-through rate. As a baseline, we depict the impact caused by a \emph{random} attack, i.e., a scenario where the adversary chooses the targets randomly. Notably, a follow-through rate of 100\% in this figure represents the scenario where \emph{all} the drivers in the network avoid the targets. While this is not likely to happen using disinformation alone, we nevertheless show the results to illustrate the theoretical limits of the attack. Note that a follow-through rate of 100\% also corresponds to a situation where the targeted streets are physically blocked by the adversary, e.g., by staging an accident. There are three observations that can be gleaned from the distributions. First, the figure shows that regardless of the type of the attack (greedy or random) and the follow-through rate, the ride time increases for some drivers and decreases for others. This can be attributed to Braess' paradox~\cite{braess1968paradoxon}, a phenomenon where the removal of some streets from a network has the counter-intuitive effect of reducing the ride times for some drivers; there is theoretical~\cite{steinberg1983prevalence} and empirical~\cite{rapoport2009choice} evidence suggesting that this phenomenon is not uncommon. Second, the figure shows that increasing the follow-through rate increases the variance of the time-delay distributions, with more vehicles experiencing longer delays. For instance, when the follow-through rate is 25\%, around 1,100 rides suffer a delay of 60 minutes or more. However, when the follow-through rate is increased to 75\%, this number increases to about 29,300. Finally, the figure shows that regardless of the follow-through rate, the distribution tends to be skewed to the left in the case of the greedy heuristic, and to the right in the case of the random attack. For example, given a follow-through rate of 50\%, the greedy heuristic results in 9.28\% of the drivers experiencing a delay, and 7.33\% a speed up of their journeys. However, for the random attack, 5.49\% experience a delay while 5.57\% experience a speed up. This underscores the importance of the strategy that the adversary uses to select targets in the network.

\section*{Discussion}

We have demonstrated how disinformation can be weaponized by an adversary to influence drivers' decisions, thereby causing major traffic disruptions at a city scale. Specifically, we considered a scenario where an adversary spreads false information about a purported accident or heavy congestion in specific locations in the city, with the aim of redirecting traffic and creating congestion in nearby areas.
Further, we analyzed two separate cases where the targeted locations were either (i) spread across the city, or (ii) confined to a particular neighborhood, and showed that the latter could be used to focus the disruption in critical neighborhoods of the city such as its business district.
A number of vulnerabilities in traffic networks have already been considered in the literature, ranging from hacking traffic signals~\cite{chen2018exposing, laszka2016vulnerability}, to defacing traffic signs to fool the image recognition algorithms used by self-driving cars~\cite{eykholt2018robust}. Our study contributes to this literature by highlighting a new vulnerability---the possibility of using disinformation to manipulate drivers' behavior at scale.

While our study is the first to consider behavioral manipulation attacks on traffic networks, it nonetheless has a number of limitations. 
In particular, we ran surveys and simulations to understand peoples' propensities to follow-through and forward fake traffic alerts, and assessed the impact that such disinformation can have on traffic. Admittedly, this approach is not ideal, since people's responses in the survey may not accurately reflect their behavior in reality. 
Another limitation is that we focus on shortest paths when computing the routes of the rides, which is not very realistic since people tend to choose routes that minimize travel time as opposed to travel distance. A more accurate alternative would be to compute fastest paths and update them periodically to take into consideration the constantly changing state of traffic. In graph theoretic terms, instead of assigning to every edge a weight that represents the length of the corresponding street, the weight can represent the expected time of traversing that street. However, this makes the model significantly more complex since the weights may change in every one of the 86,400 time steps in our simulation, and with each change, we may need to recalculate the affected routes. This is particularly challenging since there are over $5\times 10^6$ rides in our simulation. Even if we recalculate the routes only once every minute (i.e., every 60 time steps), then given the average ride duration in our simulation (which is 22.59 minutes, or 1,355 time steps), we would still need to perform over $10^8$ recalculations to simulate 24 hours of traffic in a network consisting of nearly 80,000 nodes and 235,000 edges. This makes running even a single simulation computationally demanding, and the situation becomes worse when the analysis requires hundreds of simulations, as is the case in Fig.~\ref{fig:traffic-map} and~\ref{fig:delay-detour}. By focusing on distance, we significantly reduce the computational overhead, while making sure that the generated paths: (i) have an average duration that closely matches the value observed in real life; (ii) reproduce the observed traffic frequency in different locations across the city; and (iii) reproduce the temporal distribution of the traffic intensity throughout the day; see Methods for more details.
One may consider a completely different approach; instead of a simulation-based study, one may perform a field experiment where actual drivers are sent fake notifications, and then observe and analyze the resultant impact on traffic flows. However, the main challenge in performing such an experiment would be the ethical and safety considerations involved in manipulating such a large number of unwitting drivers.

Our study has a number of implications. First, as suggested by the high propensities reported in our survey, people seem likely to follow-through on fake traffic notifications. Perhaps the reason behind such high propensities is the seemingly harmless nature of these alerts. After all, the alerts do not attempt to extract any information from the recipient, e.g., by asking them to click on suspicious external links, unlike the case with phishing attacks and spam.
Second, given the disruptive impact of the attack, it is imperative to detect and effectively counter such disinformation. While it may be difficult for individuals not physically present near a purported traffic incident to judge its veracity, one plausible solution would be to crowdsource this verification process to those who happen to be close to it. Crowdsourced fact-checking has already been shown to be very effective in identifying more vs. less reliable news sources~\cite{pennycook2019fighting}. However, out of all commonly-used navigation applications, only one offers such a functionality, namely Waze~\cite{WazeAttack1}, which serves about 11\% of all users in the US~\cite{waze_marketshare}. One clear policy implication of our study is to extend this functionality to all other navigation applications. Nevertheless, this will not mitigate an attack in which the adversary physically obstructs strategically-chosen streets, e.g., by staging accidents, in order to maximize the traffic disruption, as we have shown in the case of 100\% follow-through rate.

In conclusion, we have shown that in the age of disinformation, an adversary no longer needs to tamper with traffic control systems, but can instead focus entirely on manipulating the drivers themselves to disrupt traffic in a city. More broadly, our study provides a new perspective on the security of critical infrastructure, demonstrating that vulnerabilities can emerge not only from the hardware and software, but also from the behavior of the individuals interacting with the system.

\section*{Methods}
\begin{definition}[$f^*$]
\label{def:avg-time-evaluation-function}
Given a graph $G$, a set of rides $R$, and a traffic model $M$, the objective function $f^*$ is computed as:
$$
f^*(G,R,M) = \frac{1}{|R|} \sum_{r_i \in R} \frac{1}{\mathcal{T}(r_i,G,M)}
$$
where $\mathcal{T}(r_i,G,M)$ is the time taken to complete the ride $r_i \in R$ in the graph $G$ according to the model $M$.
\end{definition}

\begin{definition}[$M^\varnothing$]
\label{def:simple-model}
Given a graph $G$, and a set of rides $R$ where every ride $r_i \in R$ travels from a starting node $w^{start}_i \in V$ to a destination node $w^{end}_i \in V$, the time of travel according to the simple traffic model $M^\varnothing$ is:
$$
\mathcal{T}(r_i,G,M^\varnothing) = d_G(w^{start}_i,w^{end}_i)
$$
where $d_G(w^{start}_i,w^{end}_i)$ is the number of edges on a shortest path between the two nodes, unless there exists no path between them, in which case $d_G(w^{start}_i,w^{end}_i)=\infty$.
\end{definition}

\medskip
\noindent \textbf{Road network generation.} 
We obtained the road network data of Chicago from OpenStreetMap (OSM)~\cite{OpenStreetMap}. However, we could not directly utilize this data for our simulations since: (i)~some parts of the network were either disconnected or weakly connected due to the fact that OSM is crowdsourced and some streets were left unreported; (ii)~the data did not contain information about the number of lanes in each edge of the network, which is required in our traffic model $M^*$; and (iii)~a section of a road may be represented by multiple edges in OSM instead of a single edge, making the corresponding graph $G$ needlessly large, thereby increasing the processing time in our simulations.
Based on these observations, we developed an algorithm that takes the OSM data as input and generates a road network that addresses all of the aforementioned issues; see Supplementary Note~2 for more details.

\medskip
\noindent \textbf{Ride generation.} 
We generated the set of rides $R$ in Chicago by combining: (i)~publicly reported data about the number of vehicles passing by different locations in Chicago in different days of the year~\cite{chicago2019average}, with (ii)~the daily traffic intensity distribution of Chicago provided by the Texas A\&M Transportation Institute~\cite{tti2019urban}. We considered the number of vehicles corresponding to the month of October as this month had the greatest number of data points. As for the daily traffic intensity, we use an average taken over all weekdays. 
The detailed algorithm for generating the rides is presented in Supplementary Note~2.

\medskip
\noindent \textbf{Our model of traffic \textit{M}*.}
Our model of traffic, $M^*$, is a modified version of the Nagel–Schreckenberg model~\cite{nagel1992cellular}. We had to modify this model since it was only designed to model traffic in a single street, whereas our requirements call for modeling the traffic flows in a directed network of streets. We note that ours is not the first work to extend the Nagel–Schreckenberg model to a generic network. A similar extension was proposed by Gora~\cite{gora2009traffic} whose study focused on the role of traffic lights in managing the traffic flows. However, their traffic model was presented in rather vague terms, and therefore could not be used in our study.

The model $M^*$ is a cellular-automata model, and takes as input a directed road network $G=(V,E)$ and the set of rides $R=\{r_1,\ldots,r_{|R|}\}$. Each ride $r_i$ is of the form $(w^{start}_i,w^{end}_i,\theta_i)$, where $w^{start}_i \in V$ is the start node, $w^{end}_i \in V$ is the end node, and $\theta_i$ is the time of day when the ride starts. Each edge $e\in E$ in the network has a specified length $d_e$ and number of lanes $l_e$. Every lane of the edge $e$ is further divided into a number of cells; this number is denoted by $c_e$ and is computed as: $c_e=\ceil{\frac{d_e}{d_{\mathit{vehicle}}}}$, where $d_\mathit{vehicle}$ is the average length of the space occupied by each vehicle on the road, including the separation between two consecutive vehicles on the same lane. The model proceeds in discrete time steps. The time step corresponding to any given start time, $\theta_i$, is denoted by $\tau(\theta_i)$. In other words, $\tau$ maps any given clock time to a particular time step in the model. Each cell can be occupied by at most one vehicle in any given time step. Similar to the Nagel–Schreckenberg model, the speed of each ride, $v_i$, in our model is expressed as the number of cells it can traverse in a single time step. The maximum speed of all the rides is denoted by $v_{max}$.
At the beginning of each time step, rides that start at the corresponding time of the day are introduced. For each such new ride, the shortest path in the network from its start node to the end node is calculated, and the state of the ride is set as ``Waiting at the first edge of the shortest path'', i.e., waiting to be inserted into the first cell of one of the lanes in that edge, and its speed is set to zero. 
Subsequently, the model iterates over all rides that are waiting in the current time step, and if the edge that the ride is waiting at, $e$, contains a lane with an empty first cell, the ride is assigned to the first cell of this lane. In case there are multiple lanes with an empty first cell, one of them is chosen uniformly at random. The state of such a ride is then set to ``Traversing the edge $e$''.
Next, the model iterates over all the rides traversing an edge in the network, and performs the following four steps as defined by the Nagel–Schreckenberg model for each. First, the speed of the ride is increased by 1, up to the maximal speed limit $v_{max}$. 
Second, we determine if the ride can move forward on its current lane while maintaining its current speed, e.g., if its speed is 3, then, we check if there are 3 empty cells in front of it. If the number of empty cells is smaller than the speed, the ride is changed to another lane that does offer the necessary empty cells. If no such lane exists, then the ride is changed to the lane that has the maximum number of empty cells, and the speed is reduced to match the available number of empty cells in that lane.
Third, with a fixed probability $p_{slow}$, the speed of the ride is decreased by 1 to simulate random events.
Finally, the ride is moved forward by the number of cells equal to its speed. If the ride reaches the end of the edge that it is currently on, it is removed from the edge, and its status is set to either ``Waiting at the next edge of the shortest path'' (if the ride has more edges to traverse) or to ``Finished'' (if the ride reached its final destination). The process continues until a number of time steps $t_{max}$ are completed and all the rides in $R$ reach their destination. For more details and the pseudocode for the model, see Supplementary Note~3.

In our study, we assume that each time step in the model corresponds to 1 second, i.e., $\tau(\theta_i) = \theta_i$. 
We further set $d_{vehicle} = 7.5$ metres and $v_{max}=5$, which are the standard values used in the Nagel–Schreckenberg model~\cite{schadschneider1999nagel}.
These parameters yield a maximum speed of 84~mph, since a vehicle can in 1~second traverse up to 5 cells, each of which is 7.5~m long.
The probability that a ride slows down is taken as $p_{slow}=\frac{1}{100}$.
Finally, we set $t_{max}=$ 86,400, which is the number of seconds in 24 hours.
To assess whether our model generates realistic rides, we ran 100 simulations of traffic in the Chicago network, with each simulation considering a new set of rides generated according to the procedure described above. We found that the average time of travel in the network is 22.59 minutes. In comparison, the American Driving Survey~\cite{kim2019american} reported that a typical American driver in 2016 and 2017 spent 51 minutes driving each day, making 2.2 trips, resulting in an average ride time of 23.18 minutes. This close correspondence between the real data and simulation outcomes demonstrates that our model is indeed realistic.

\medskip
\noindent \textbf{Traffic simulations under attack.}
We determine the set of edges that are targeted by the adversary. This depends on the budget $b$ of the attack as well as on whether the adversary uses the greedy heuristic or adopts the random approach. The targets are chosen from a set $Q\subseteq E$. If the adversary is choosing targets from across the entire city, then $Q=E$; otherwise, if the adversary is targeting a particular neighborhood (e.g., Chicago Loop in our scenario), then $Q$ consists of all edges in that neighborhood. The set of chosen targets is then denoted by $Q^*\subseteq Q$.
Next, the number of drivers who follow-through, i.e., those who modify their rides to avoid the targets, is determined based on the disinformation follow-through rate. The specific rides that follow-through are then chosen randomly.
The route of every such ride is then computed as a shortest path in the graph $(V,E\setminus Q^*)$, thereby ensuring that the ride avoids all edges in $Q^*$, i.e., avoids all streets targeted by the adversary. As for the remaining rides, their route is computed as a shortest path in $(V,E)$, without being constrained by the targets.
We perform $100$ simulations for each attack strategy (either greedy or random) and follow-through rate, using the same set of rides generated for the baseline simulation corresponding to the scenario where no attack is carried out.

On a side note, when illustrating the network in Fig.~\ref{fig:traffic-map}a and \ref{fig:traffic-map}b, to improve visibility we omitted the section corresponding to O'Hare airport, which is at the periphery of the network; see Fig.~S2 in Supplementary Materials for the full network. Note that our simulations were run on the entire network, and the section was only omitted in the figures without loss of information since the attacks caused no traffic changes in this area.

\section*{Ethics statement}
The research was approved by the Institutional Review Boards of the New York University Abu Dhabi and the National University of Singapore.

\bibliographystyle{unsrt}
\bibliography{bibliography-traffic}

\clearpage
\appendix

\renewcommand{\thesection}{S\arabic{section}}
\renewcommand{\thefigure}{S\arabic{figure}}
\renewcommand{\thetable}{S\arabic{table}}
\renewcommand{\thedefinition}{S\arabic{definition}}

\section{Supplementary note 1: Survey questionnaire}
\label{sec:note1}
We ran a survey to understand how participants react to fake notifications regarding traffic alerts. We recruited 3,301 participants through Amazon Mechanical Turk, who were then directed to a survey on the Qualtrics platform. The participants were residents of one of the 35 most traffic congested cities in the United States. The survey consisted of five steps, an overview of which is provided below. Detailed explanations are provided in Sections~\ref{sec:consent:form} through \ref{sec:closingmessage}.

\begin{itemize}[leftmargin={1.5cm}]
\item[\textbf{Step 1:}] A consent form is displayed (see Section~\ref{sec:consent:form}), after which the demographic questions are asked (see Section~\ref{sec:demographic:questions});

\item[\textbf{Step 2:}] The following sentence is displayed:
    \begin{itemize}
    \item ``\textit{Suppose you receive the following message from SMSAlert}'';
    \end{itemize}

\item[\textbf{Step 3:}] The message containing the fake notification is displayed (see Section~\ref{sec:messages:fake:notifications});

\item[\textbf{Step 4:}] The participant is asked questions regarding their follow-through and forwarding propensities (see Section~\ref{Question_subsection});

\item[\textbf{Step 5:}] The participant is awarded a financial compensation of \$0.5, and thanked for their participation (Section~\ref{sec:closingmessage}).
\end{itemize}

\subsection{Consent form}\label{sec:consent:form}
The consent form consists of the following text:
\ \\
\begin{changemargin}{1.1cm}{1cm}
\item \textit{Welcome to this study investigating how humans behave upon receiving text notifications. You are eligible to participate in the study at this time if you are:}
\end{changemargin}

\begin{enumerate}[leftmargin={2cm}, rightmargin={2cm}]
    \item \textit{18 years of age or older;}
    \item \textit{live in one of the following 35 cities: Albuquerque (NM), Austin (TX), Baltimore (MD), Boston (MA), Charlotte (NC), Chicago (IL), Columbus (OH), Dallas (TX), Denver (CO), Detroit (MI), El Paso (TX), Fort Worth (TX), Fresno (CA), Houston (TX), Indianapolis (IN), Jacksonville (FL), Las Vegas (NV), Los Angeles (CA), Louisville (KY), Memphis (TN), Milwaukee (WI), Nashville (TN), New York City (NY), Oklahoma City (OK), Philadelphia (PA), Phoenix (AZ), Portland (OR), Sacramento (CA), San Antonio (TX), San Diego (CA), San Francisco (CA), San Jose (CA), Seattle (WA), Tucson (AZ), or Washington (DC).}
\end{enumerate}

\begin{changemargin}{1.1cm}{1cm}
\item \textit{The questionnaire asks about your background and your reaction upon receiving text notifications. This survey is anonymous, i.e., it does not contain individually identifiable data from you. Your participation is voluntary, and you may close the survey at any point. }
\ \\ \\
\textit{The questionnaire is expected to last on average 5 minutes. An amount of \$0.5 will be paid upon successful completion of the survey.}
\ \\ \\
\textit{This research is conducted by Jimmy Chih-Hsien Peng and Gururaghav Raman at the National University of Singapore, as well as Talal Rahwan, Bedoor AlShebli, and Marcin Waniek at New York University  Abu Dhabi, and has been approved by the respective Institutional Review Boards.  For questions about the rights of research participants, you may contact the University Committee on Activities Involving Human Subjects, New York University Abu Dhabi, irbnyuad@nyu.edu and the National University of Singapore, irb@nus.edu.sg.
\ \\ \\
If you have any questions, suggestions or concerns, please feel free to reach out to us at nyuad.textnotification@nyu.edu – an email address that only researchers associated with this project have access to.
\ \\ \\
Please do not complete the survey more than once. Upon finishing the survey you will receive a completion code. The payment of \$0.5 will be made once you've entered that code in the space provided. \textbf{Please do not close the browser with your MTurk account}. 
If you read this consent form, and would like to participate in this study, press the button below!}
\end{changemargin}

\subsection{Demographic questions}\label{sec:demographic:questions}
The demographic questions that were asked to each participant are detailed below:
\begin{enumerate}
    \item \textit{What is the sex listed on your birth certificate?}
\textit{
        \begin{enumerate}
            \item Male
            \item Female
        \end{enumerate}
    \item What is your ethnicity?
        \begin{enumerate}
            \item Hispanic
            \item Non-Hispanic
        \end{enumerate}
    \item What is your race?
        \begin{enumerate}
            \item White
            \item Black or African American
            \item Asian
            \item Native American
            \item Middle Eastern or North African
            \item Mixed
            \item Other
        \end{enumerate}
    \item What is your age?} 
    \begin{itemize}
    \item[]\hspace*{-0.7cm}[empty field to be filled by a number]
    \end{itemize}
\textit{
    \item What is your highest completed level of education?
        \begin{enumerate}
            \item Less than high school
            \item High school graduate or equivalent (e.g., GED)
            \item Some college
            \item 2 year degree (i.e. Associate's degree)
            \item 4 year degree (i.e. Bachelor's degree)
            \item Masters or Professional degree (i.e. MBA, MPP, etc)
            \item Doctoral Degree
        \end{enumerate}
    \item What best describes your employment situation?
        \begin{enumerate}
            \item Full-time employed
            \item Part-time employed
            \item Unemployed
            \item Caregiver (e.g., children, elderly) or homemaker
            \item Retired
            \item Full-time student
            \item Other
        \end{enumerate}
        \item What was your yearly personal income in 2018 (include salary, interests, returns on investments, etc)?
            \begin{enumerate}
                \item Less than \$10,000
                \item \$10,000-\$19,999
                \item \$20,000-\$29,999
                \item \$30,000-\$39,999
                \item \$40,000-\$49,999
                \item \$50,000-\$59,999
                \item \$60,000-\$69,999
                \item \$70,000-\$79,999
                \item \$80,000-\$99,999
                \item \$100,000-\$119,999
                \item
               \$120,000-\$149,999
               \item \$150,000-\$199,999
               \item \$200,000 - more
           \end{enumerate}
        \item Which city do you live in?
}
        \begin{itemize}
            \item \textit{Albuquerque, NM}
            \item \textit{Austin, TX}
            \item \textit{Baltimore, MD}
            \item \textit{Boston, MA}
            \item \textit{Charlotte, NC}
            \item \textit{Chicago, IL}
            \item \textit{Columbus, OH}
            \item \textit{Dallas, TX}
            \item \textit{Denver, CO}
            \item \textit{Detroit, MI}
            \item \textit{El Paso, TX}
            \item \textit{Fort Worth, TX}
            \item \textit{Fresno, CA}
            \item \textit{Houston, TX}
            \item \textit{Indianapolis, IN}
            \item \textit{Jacksonville, FL}
            \item \textit{Las Vegas, NV}
            \item \textit{Los Angeles, CA}
            \item \textit{Louisville, KY}
            \item \textit{Memphis, TN}
            \item \textit{Milwaukee, WI}
            \item \textit{Nashville, TN}
            \item \textit{New York City, NY}
            \item \textit{Oklahoma City, OK}
            \item \textit{Philadelphia, PA}
            \item \textit{Phoenix, AZ}
            \item \textit{Portland, OR}
            \item \textit{Sacramento, CA}
            \item \textit{San Antonio, TX}
            \item \textit{San Diego, CA}
            \item \textit{San Francisco, CA}
            \item \textit{San Jose, CA}
            \item \textit{Seattle, CA}
            \item \textit{Tucson, AZ}
            \item \textit{Washington, DC}
            \item \textit{Other:} [text box appears]\\ 
            (if ``Other" selected, participant is not eligible and is thanked and exits the survey)
        \end{itemize}
\end{enumerate}

\subsection{Message with the fake notification}\label{sec:messages:fake:notifications}
The following sentence was displayed to the participant:
\begin{itemize}
\item ``\textit{Suppose you receive the following message from SMSAlert}''
\end{itemize}
\noindent Below this text, the message shown in Fig.~\ref{fig:message:screenshots} was displayed to the participant.

\begin{figure*}[ht]
	\centering	
    \includegraphics[width=0.35\linewidth]{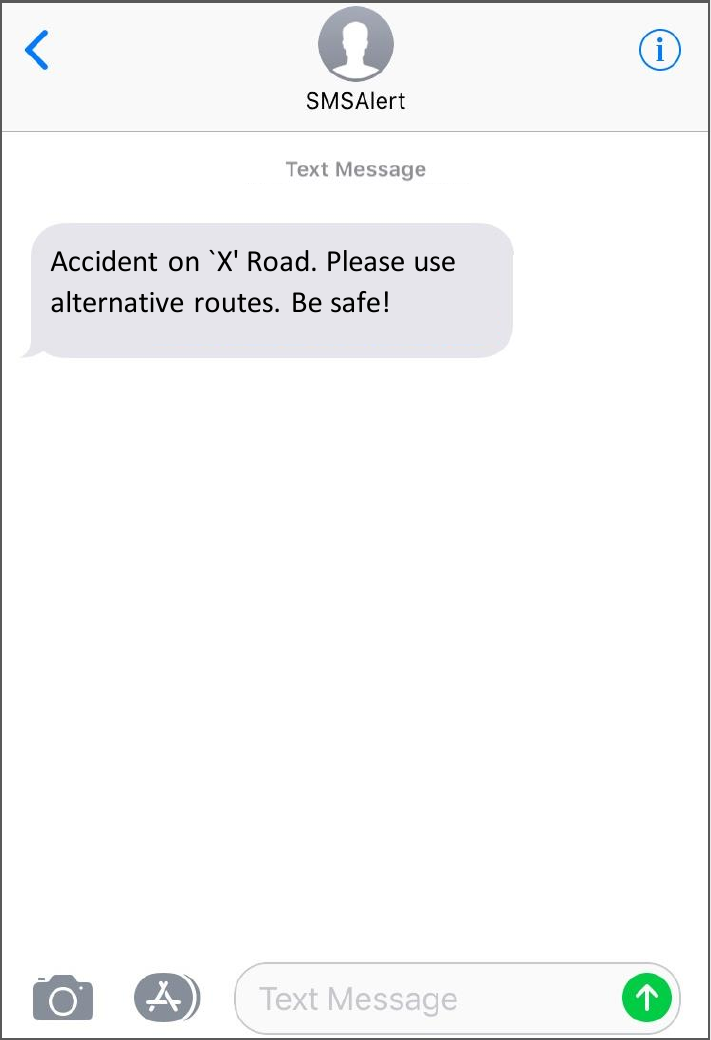}
	\caption{The message shown to the participants in our survey.}
\label{fig:message:screenshots}
\end{figure*}

\subsection{Questions to assess propensities to follow-through \& forward notifications}
\label{Question_subsection}
\begin{enumerate}
    \item \textit{What is the likelihood that you use an alternative route?}
    \begin{itemize}[leftmargin={0.4cm}]
        \item[] [Likert scale with values: $0,1,\ldots,10$, where ``0" is labeled at ``never" and ``10" is labeled as ``definitely"]
    \end{itemize}
    \item \textit{What is the likelihood that you forward this message to your friends?}
    \begin{itemize}[leftmargin={0.4cm}]
        \item[] [Likert scale with values: $0,1,\ldots,10$, where ``0" is labeled at ``never" and ``10" is labeled as ``definitely"]
    \end{itemize}
\end{enumerate}

\subsection{Closing message and compensation}
\label{sec:closingmessage}
\textit{This is the end of the study. Thank you for your participation. Please make note of the following 7-digit code. You will input it through Mechanical Turk to indicate your completion of the study. \textbf{Then click the button on the bottom of this page to submit your answers. You will not receive credit unless you click this button.} 
\ \\
\ \\
7-Digit Code: }[code appears]

\clearpage

\section{Supplementary note 2: Road network and ride generation}
\label{sec:note2}
This section details how the road network and vehicle rides in Chicago were generated.

\subsection{Road network generation}
\label{app:network-generation}
We obtained the road network data of Chicago from OpenStreetMap (OSM)~\cite{OpenStreetMap}. However, we could not directly utilize this data for our simulations since: (i)~some parts of the network were either disconnected or weakly connected due to the fact that OSM is crowdsourced and some streets were left unreported; (ii)~the data did not contain information about the number of lanes in each edge of the network, which is required in our traffic model; and (iii)~a section of a road may be represented by multiple edges in OSM instead of a single edge, making the corresponding graph needlessly large, thereby increasing the processing time in our simulations.
Based on these observations, we developed an algorithm that takes the OSM data as input and generates a road network that addresses all of the aforementioned issues.
The steps of the algorithm are as follows:

\begin{enumerate}
\item Create the set of nodes of the road network by extracting from OSM all nodes belonging to ways within a given area such that the value of the way's \textit{highway} key is one of the following: \textit{motorway}, \textit{trunk}, \textit{primary}, \textit{secondary}, \textit{tertiary}, \textit{unclassified}, \textit{residential}, or \textit{service}.
For each pair of nodes connected with a way in the OSM data, connect them in the road network with a single directed edge if the way's set of tags contains the key \textit{oneway} with the value \textit{yes}, \textit{true} or \textit{1}; otherwise, connect them with directed edges in both directions.
For each such created edge, record its length defined as the geographical distance between the coordinates of the two nodes, computed using the haversine method~\cite{van2012heavenly}.
Moreover, set the number of lanes of the edge based on the value of the way's \textit{highway} key as follows:
\begin{itemize}
\item For \textit{motorway} or \textit{trunk}, set the number of lanes to $4$;
\item For \textit{primary} or \textit{secondary}, set the number of lanes to $3$;
\item For \textit{tertiary} or \textit{unclassified}, set the number of lanes to $2$;
\item For \textit{residential} or \textit{service}, set the number of lanes to $1$.
\end{itemize}

\item Merge nearby nodes---defined here as being within 20 meters of each other---into a single node.
Let $G'=(V',E')$ be the network before merging any nodes, where $V'$ is the set of nodes and $E'$ is the set of edges, and let $G=(V,E)$ be the network after merging the nodes.
For each node $x \in V$ created by merging a group of nodes $X \subset V'$, set its coordinates to the average of the coordinates of all the nodes in $X$.
Now, suppose we merged a group of nodes $X \subset V'$ into a single node $x \in V$, and merged another group $Y \subset V'$ into a node $y\in V$.
Then, we add an edge $(x,y)\in E$ if and only if $\exists v \in X, u \in Y: (v,u)\in E'$, in which case we set the length of the edge to the geographical distance between the coordinates of $x$ and $y$, and set the number of lanes of the edge to the maximum out of all the edges in $X \times Y$.

\item Ensure strong connectivity of the network by repeating the following process:
\begin{enumerate}
\item Select two strongly connected components $G_1=(V_1,E_1)$ and $G_2=(V_2,E_2)$ from the road network such that the geographical distance between the closest pair of nodes $x \in V_1$ and $y\in V_2$ is minimal;
\item Out of the edges $(x,y)$ and $(y,x)$, add the missing one(s) to $E$. Note that at least one of them is missing from $E$, because otherwise $G_1$ and $G_2$ would be a single strongly connected component.
\item Set the length of both $(x,y)$ and $(y,x)$ to the geographical distance between $x$ and $y$;
\item If one of the edges $(x,y)$ or $(y,x)$ existed before step (b), set the number of lanes of the other edge to the number of lanes of the existing edge; otherwise set the number of lanes of both $(x,y)$ and $(y,x)$ to $1$.
\end{enumerate}

This process is repeated until the road network becomes strongly connected.
Note that this step is necessary since the OSM road network is often disconnected or weakly connected (due to the fact that OSM is crowdsourced and some roads may be left unreported), whereas the real road network of a city is strongly connected.

\item To make the computation more efficient, contract the edges of the road network, \ie, iteratively remove every node $y$ that satisfies one of the following conditions:
\begin{itemize}
\item Has exactly one predecessor $x$ and exactly one successor $z$, such that $x \neq z$;
\item Has exactly two predecessors, which are also its only two successors.
\end{itemize}
For example, if node $y$ is connected only to nodes $x$ and $z$, then replace the edges $(x,y)$ and $(y,z)$ with the edge $(x,z)$.
The length of the new edge is set to be the sum of lengths of $(x,y)$ and $(y,z)$, while the number of lanes of the edge is set to be the maximum of the numbers of lanes of $(x,y)$ and $(y,z)$.
\end{enumerate}

\noindent Fig.~\ref{fig:chicago-osm} presents the road network of Chicago that was generated using the above algorithm while considering the OSM relation \#122604 as the area boundaries.

\begin{figure}[tb]
\centering
\includegraphics[width=.7\linewidth]{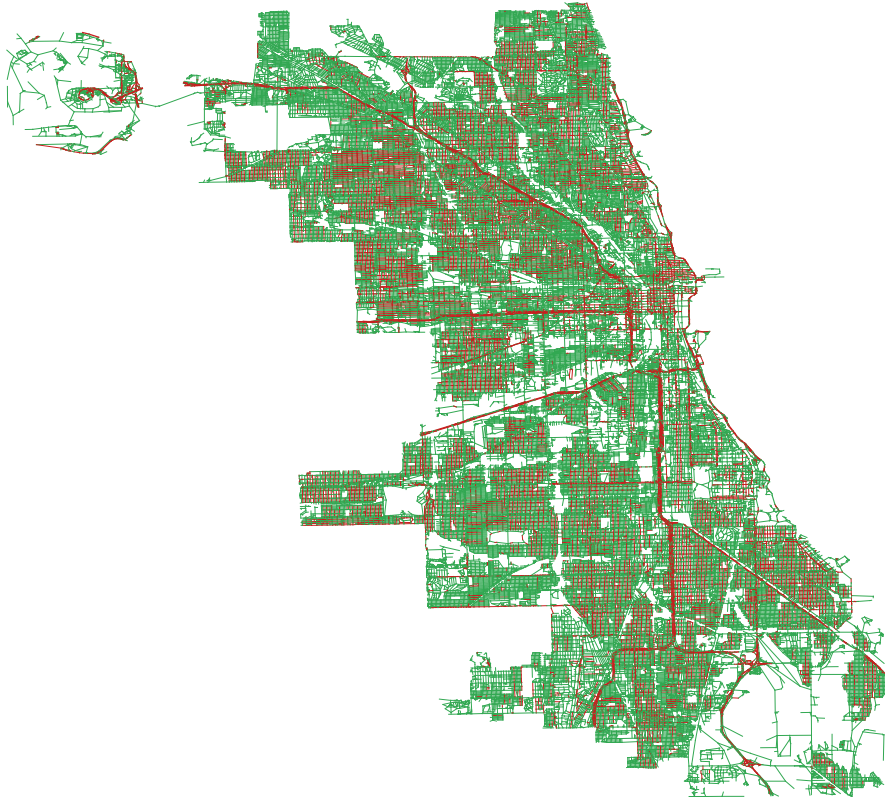}
\caption{The road network of Chicago generated by our algorithm based on OSM data. Green edges represent two-way streets, while red edges represent one-way streets.}
\label{fig:chicago-osm}
\end{figure}

\subsection{Ride generation}
\label{app:car-routes-generation}
We now describe the algorithm that we developed to generate vehicle rides for our traffic simulations. The input to the algorithm is the following:
\begin{itemize}
    \item The directed road network $G=(V,E)$ generated using the algorithm described in Section~\ref{app:network-generation}, where each edge has an assigned length and number of lanes;
    \item Data about the city traffic in the form of a set of pairs $D = \{(c_1,k_1), \ldots, (c_{|D|},k_{|D|})\}$, where $c_i$ is a location (represented by its geographical coordinates), and $k_i$ is the daily average of the number of vehicles passing by the location $c_i$;
    \item Distribution $\Theta$ of the intensity of the traffic throughout the day, measured by the number of rides that have started but have not yet reached their destination.
\end{itemize}
For the purpose of simulating the traffic in Chicago, $G$ is generated taking the OSM relation \#122604 as the area boundaries, $D$ is generated using the data for October (the month with the greatest number of data points) provided by the city of Chicago~\cite{chicago2019average}, while $\Theta$ is based on the data provided by the Texas A\&M Transportation Institute for Chicago~\cite{tti2019urban} (we use the average taken over all weekdays).

The output of the algorithm is the set of rides $R$, where each ride is characterized by:
\begin{itemize}
\item Start node, $w^{start} \in V$;
\item End node, $w^{end} \in V$;
\item Time of day when the journey starts, $\theta$.
\end{itemize}

Algorithm~\ref{alg:generating-rides} presents the high-level pseudocode for generating the rides.
The detailed steps are as follows:

\begin{enumerate}
\item \label{pt:compute-x} 
\textbf{Initialize the number of vehicles passing through different nodes:}
Let $D(v_i)\subseteq D$ be the set of pairs $(c_j,k_j)$ for which $v_i\in V$ is the geographically closest node to the coordinates $c_j$. More formally, $D(v_i) = \{(c_j,k_j) \in D: v_i = \argmin_{w \in V} \mathit{distance}(w,c_j)\}$. Moreover, let $X\subseteq V$ be the set of nodes that are closest to at least one set of coordinates. More formally, $X = \{v_i\in V : D(v_i)\neq \emptyset\}$. Note that for any given node $v_i\in X$, every pair $(c_j,k_j) \in D(v_i)$ indicates that the number of vehicles who passed by the location $c_j$ is equal to $k_j$. Since this data was collected via a sensor, we will refer to the node $v_i$ as a \textit{sensor}. For every such sensor $v_i \in X$, let us denote by $x_i$ the variable used in our simulation to count the number of rides who are supposed to go through $v_i$. This variable is initialized to the average value of $k_j$ taken over all $(c_j,k_j) \in D(v_i)$. More formally,
$$
x_i= \frac{\sum_{(c_j,k_j) \in D(v_i)} k_j}{|D(v_i)|}.
$$
\noindent Later on in our simulation, the value stored in any such variable $x_i$ will be reduced by 1 every time a ride passes by the node $v_i$. Here, the goal is to reach a state at which the values $x_i$ in all sensors are equal to zero, indicating that the number of vehicles who passed by $v_i$ in our simulation matches the average number of vehicles who passed by the coordinates $c_j : (c_j,k_j)\in D(v_i)$.

\item \label{pt:preprocess} \textbf{Select the subset of shortest paths $\Pi$:}
For every pair of nodes $v,w \in V$ such that $v \neq w$, compute a shortest path $\pi$ from $v$ to $w$ in $G$. If $\pi$ contains at least one sensor, \ie, if $\pi \cap X \neq \emptyset$, then with probability $p_{load} = \frac{1}{1000}$ add $\pi$ to the set $\Pi$. Paths in $\Pi$ will be the only paths used to generate the set of rides $R$. In other words, instead of sampling the rides out of all possible shortest paths in the graph, we sample them out of those in $\Pi$ to reduce the computational overhead and speed up the simulations. By sampling the rides from a subset of shortest paths, we ensure that every ride minimizes the travel distance (see Discussion in the main manuscript).

\item \label{pt:add-ride} \textbf{Add a single path $\pi^*$ to $R$ and update the data structures:}
Randomly select a shortest path $\pi^*  = (\pi^*_1,\ldots,\pi^*_{|\pi^*|}) \in \Pi$. Here, every $\pi \in \Pi$ is selected with probability proportional to its length, according to a normal distribution $\mathcal{N}$ with mean $\mu = 8$ kilometers and standard deviation $\sigma = 1$ kilometer. For every sensor $v_i \in \pi^* \cap X$ decrease the value of $x_i$ by $1$. If for any $v_i \in \pi^* \cap X$ we now have $x_i = 0$, remove from $\Pi$ every path $\pi$ that goes through $v_i$. Add the ride $(\pi^*_1,\pi^*_{|\pi^*|},\theta)$ to the set of rides $R$, where the starting time $\theta$ is selected according to the distribution $\Theta$.

\item \textbf{Check the end condition:}
Repeat step~\ref{pt:add-ride} until the generated rides match the sensor data or we run out of paths to be added, \ie, until $(\forall_{v_i \in V}x_i = 0) \lor (\Pi = \emptyset)$.
\end{enumerate}


\begin{algorithm}[t]
\caption{Ride generation.}
\label{alg:generating-rides}
\begin{algorithmic}[0]
\small
\Input{Road network $G=(V,E)$; traffic data $D = \{(c_1,k_1), \ldots, (c_{|D|},k_{|D|})\}$; traffic intensity distribution $\Theta$}
\Output{Set of rides $R$}

\State $X \gets \{v_i \in V : D(v_i) \neq \emptyset\}$ \Comment{\textcolor{blue}{\footnotesize Identify the set of nodes in $G$ that are closest to at least one set of coordinates $c_j$}}
\For{$v_i \in X$}
	\State $x_i \gets \frac{\sum_{(c_j,k_j) \in D(v_i)} k_j}{|D(v_i)|}$ \Comment{\textcolor{blue}{\footnotesize Initialize $x_i$ to the average $k_j$ in $D(v_i)$}}
\EndFor

\State $\Pi \gets \emptyset$ \Comment{\textcolor{blue}{\footnotesize Initialize the set of shortest paths from which the rides will be sampled}}
\For{$v, w \in V$} 
    \State Let $\pi$ be a shortest path from $v$ to $w$ in $G$
	\If{$\pi \cap X \neq \emptyset$} \Comment{\textcolor{blue}{\footnotesize If $\pi$ contains at least one sensor}}
		\State With probability $p_{load}$ add $\pi$ to $\Pi$
	\EndIf
\EndFor

\State $R \gets \emptyset$ \Comment{\textcolor{blue}{\footnotesize Initialize the set of generated rides}}
\While{$\exists_{v_i \in V}x_i > 0 \land \Pi \neq \emptyset$} \Comment{\textcolor{blue}{\footnotesize Repeat until the generated rides match the sensor data or we run out of paths}}
	\State Draw $\pi^*$ from $\Pi$ according to the distribution $\mathcal{N}(\mu,\sigma)$
	\For{$v_i \in \pi^* \cap X$} \Comment{\textcolor{blue}{\footnotesize Repeat for every sensor that appears on $\pi^*$}}
		\State $x_i \gets x_i - 1$ \Comment{\textcolor{blue}{\footnotesize The number of rides that are supposed to pass by $v_i$ is decreased by 1}}
		\If{$x_i = 0$} \Comment{\textcolor{blue}{\footnotesize If the generated rides match the sensor data}}
			\State $\Pi \gets \Pi \setminus \{ \pi \in \Pi : v_i \in \pi \}$ \Comment{\textcolor{blue}{\footnotesize Remove from $\Pi$ all the shortest paths that contain $v_i$}}
		\EndIf
	\EndFor	
	\State Draw $\theta$ from $\Theta$ \Comment{\textcolor{blue}{\footnotesize Draw the start time of the ride according to the daily traffic intensity distribution}}
	\State $R \gets R \cup \{(\pi^*_1,\pi^*_{|\pi^*|},\theta)\}$  \Comment{\textcolor{blue}{\footnotesize Add the generated ride $(\pi^*_1,\pi^*_{|\pi^*|},\theta)$ to the set of rides $R$}}
\EndWhile

\State \Return $R$

\end{algorithmic}
\end{algorithm}

\clearpage

\section{Supplementary note 3: Our model of traffic \textit{M*}}
\label{sec:note3}
Our model of traffic, $M^*$, is a modified version of the Nagel–Schreckenberg model~\cite{nagel1992cellular}. We had to modify this model since it was only designed to model traffic in a single street, whereas our requirements call for modeling the traffic flows in a directed network of streets. We note that ours is not the first work to extend the Nagel–Schreckenberg model to a generic network. A similar extension was proposed by Gora~\cite{gora2009traffic} whose study focused on the role of traffic lights in managing the traffic flows. However, their traffic model was presented in rather vague terms, and therefore could not be used in our study.

We now describe our model in detail. The input to our traffic model $M^*$ is the following:
\begin{itemize}
\item The directed road network $G=(V,E)$, where each edge $e\in E$ has an assigned length $d_e$ and a number of lanes $l_e$;
\item The set of rides $R=\{r_1,\ldots,r_{|R|}\}$, with each ride $r_i$ being of the form $(w^{start}_i,w^{end}_i,\theta_i)$, where $w^{start}_i \in V$ is the start node, $w^{end}_i \in V$ is the end node, and $\theta_i$ is the time of day when the ride starts.
\end{itemize}

The model proceeds in discrete time steps. The time step corresponding to any given start time, $\theta_i$, is denoted by $\tau(\theta_i)$. In other words, $\tau$ maps any given clock time to a particular time step in the model. For any edge, $e\in E$, the street associated with $e$ is divided into a number of cells; each of which can be occupied by at most one vehicle in any given time step. The number of cells, $c_e$, in edge $e$ is denoted by $c_e$ and is computed as $c_e=\ceil{\frac{d_e}{d_{\mathit{vehicle}}}}$, where $d_\mathit{vehicle}$ is the average length of the space occupied by each vehicle on the road, including the separation between two consecutive vehicles on the same lane. Similar to the Nagel–Schreckenberg model, the speed of each ride, $v_i$, in our model is expressed as the number of cells it can traverse in a single time step. The maximum speed of all the rides is denoted by $v_{max}$.

Algorithm~\ref{alg:traffic-simulation} presents a high-level pseudocode of the model. In words, for each time step $t$, the model involves the following actions:

\begin{enumerate}
\item For every ride, $(w^{start}_i,w^{end}_i,\theta_i) \in R$ whose starting time $\tau(\theta_i)=t$:
    \begin{itemize}
        \item Add the ride to the set of ongoing rides $\Omega$;
        \item Compute the shortest path $\pi_i$ from $w^{start}_i$ to $w^{end}_i$;
        \item Set the state of the ride $\gamma_i$ as waiting on the first edge of $\pi_i$;
        \item Set the speed of the ride as $v_i=0$.
    \end{itemize}

\item For every ride $(w^{start}_i,w^{end}_i,\theta_i) \in \Omega$ that is waiting on some edge $e$, if there exists a lane $\lambda$ of $e$ with an empty first cell, put the ride in the first cell ($c_i = 1$) of $\lambda$, i.e.,  set $\lambda_i = \lambda$, and set the state of the ride as traversing edge $e$. In case there are multiple lanes with an empty first cell, then one of them is chosen uniformly at random to be $\lambda_i$.

\item For every ride $(w^{start}_i,w^{end}_i,\theta_i) \in \Omega$ that is traversing some edge $e$:
    \begin{enumerate}
        \item Increase the vehicle's speed by 1 if it is smaller than $v_{max}$, \ie, $v_i = \max(v_i+1,v_{max})$.
        
        \item If the current lane, $\lambda_i$, does not have at least $v_i$ empty cells in front of the vehicle, change the lane to the one with the greatest number of empty cells, \ie,
        $$
        \lambda_i = \argmax_{\lambda \in \{1,\ldots,l_e\}} x_e(\lambda,c_i)
        $$
        where $x_e(\lambda,c_i)$ is the number of empty cells in front of cell $c_i$ on lane $\lambda$ of edge $e$.
        
        \item Match the vehicle's speed to the number of empty cells in front of it in its current lane, \ie, $v_i = \min(v_i, x_e(\lambda_i,c_i))$.
        
        \item If the vehicle's speed is greater than $0$, then with a probability $p_{slow}$ decrease it by 1, \ie, $v_i = \max(v_i-1, 0)$.
        
        \item Move the vehicle forward by $v_i$ cells on its current lane, \ie, $c_i = c_i + v_i$. If the vehicle reached the last cell of $e$, remove it from $e$. If the vehicle got removed from $e$, either set its state to waiting at the next edge in $\pi_i$ (if there are more edges in $\pi_i$), or remove it from the set of ongoing rides $\Omega$ (if it reached its final destination).
    \end{enumerate}
\end{enumerate}

The process is continued until $t \geq t_{max}$ and all vehicles finish their journeys.  

\begin{algorithm}[tbh!]
\caption{Model of traffic $M^*$.}
\label{alg:traffic-simulation}
\begin{algorithmic}[0]

\small

\Input{Road network $G=(V,E)$ where each edge $e \in E$ has an assigned length $d_e$ and number of lanes $l_e$; a set of rides $R=\{(w^{start}_1,w^{end}_1,\theta_1),\ldots,(w^{start}_{|R|},w^{end}_{|R|},\theta_{|R|})\}$}

\State $t \gets 0$ \Comment{\textcolor{blue}{\footnotesize Set the current time step to zero}}
\State $\Omega \gets \emptyset$ \Comment{\textcolor{blue}{\footnotesize Initialize the set of ongoing rides}}

\While{$t < t_{max} \lor \Omega \neq \emptyset$}

\For{$(w^{start}_i,w^{end}_i,\theta_i) \in R : \tau(\theta_i) = t$}\Comment{\textcolor{blue}{\footnotesize For every ride that starts at time $t$}}
	\State $\Omega \gets \Omega \cup \{(w^{start}_i,w^{end}_i,\theta_i)\}$
	\State Let $\pi_i$ be the shortest path from $w^{start}_i$ to $w^{end}_i$ in $G$
	\State $\gamma_i \gets \text{Waiting at}~e_1(\pi_i)$ \Comment{\textcolor{blue}{\footnotesize Set the status of the ride to ``Waiting at the first edge of $\pi_i$''}}
	\State $v_i \gets 0$ \Comment{\textcolor{blue}{\footnotesize Set the speed of the ride to zero}}
\EndFor

\For{$(w^{start}_i,w^{end}_i,\theta_i) \in \Omega : \gamma_i = \text{Waiting at}~e_j (\pi_i)$} \Comment{\textcolor{blue}{\footnotesize For every ride waiting at an edge in the network}}
	\If{there exists a lane $\lambda$ in $e$ with empty first cell}
		\State $\lambda_i \gets \lambda$ \Comment{\textcolor{blue}{\footnotesize Set the current lane of the ride to $\lambda$}}
		\State $c_i \gets 1$ \Comment{\textcolor{blue}{\footnotesize Set the ride's location to the first cell of $\lambda_i$}} 
		\State $\gamma_i \gets \text{Traversing}~e_j (\pi_i)$ \Comment{\textcolor{blue}{\footnotesize Set the status of the ride to ``Traversing the edge $e_j$''}}
	\EndIf
\EndFor

\For{$(w^{start}_i,w^{end}_i,\theta_i) \in \Omega : \gamma_i = \text{Traversing}~e_j (\pi_i)$} \Comment{\textcolor{blue}{\footnotesize For every ride traversing an edge in the network}}
    \State $e \gets e_j (\pi_i)$ \Comment{\textcolor{blue}{\footnotesize Let $e$ denote the edge that the ride is traversing}}
	\State $v_i \gets \max(v_i+1,v_{max})$ \Comment{\textcolor{blue}{\footnotesize Increase speed of the ride by 1 unless it has reached the maximum speed}}
	\If {$x_e(\lambda_i,c_i) < v_i$} \Comment{\textcolor{blue}{\footnotesize If there are less than $v_i$ empty cells in front of the ride}}
		\State $\lambda_i \gets \argmax_{\lambda \in \{1,\ldots,l_e\}} x_e(\lambda,c_i)$  \Comment{\textcolor{blue}{\footnotesize Reassign the ride to the lane with the most empty cells ahead}}
	\EndIf 
	\State $v_i \gets \min(v_i, x_e(\lambda_i,c_i))$ \Comment{\textcolor{blue}{\footnotesize Cap the speed to the number of empty cells}}
	\State $v_i \gets \max(v_i-1, 0)$ with probability $p_{slow}$ \Comment{\textcolor{blue}{\footnotesize With a probability of $p_{slow}$, slow down the ride by 1}}
	\State $c_i \gets c_i + v_i$ \Comment{\textcolor{blue}{\footnotesize Move the ride ahead by $v_i$ cells}}
	\If{$c_i=c_e$} \Comment{\textcolor{blue}{\footnotesize If the ride has reached the last cell of the edge $e$}}
		\If{there exists an edge $e^*$ after $e$ in $\pi_i$ } \Comment{\textcolor{blue}{\footnotesize If the ride has not reached the end of $\pi_i$}}
			\State $\gamma_i \gets \text{Waiting at}~e^*$ \Comment{\textcolor{blue}{\footnotesize Set the status of the ride to ``Waiting at edge $e^*$''}}
		\Else
			\State $\Omega \gets \Omega \setminus \{(\pi_i,\theta_i)\}$ \Comment{\textcolor{blue}{\footnotesize Remove the ride from the set of ongoing rides}}
			\State $\gamma_i \gets \text{Finished}$ \Comment{\textcolor{blue}{\footnotesize Set the status of the ride to ``Finished''}}
		\EndIf
	\EndIf
\EndFor

\State $t \gets t + 1$ \Comment{\textcolor{blue}{\footnotesize Increase the time step by 1}}

\EndWhile

\end{algorithmic}
\end{algorithm}

\clearpage

\section{Supplementary note 4: Definitions}
\label{sec:note4}

\begin{definition}[The problem of Maximizing Disruption]
\label{def:traffic-attack}
The problem is defined by a tuple, $(G,Q,R,f,M,b)$, where $G=(V,E)$ is a directed network, $Q \subseteq E$ is the set of edges that the adversary can choose from, $R$ is the set of rides, $f$ is an objective function representing traffic quality, $M$ is a model of traffic, and $b \in \N$ is the budget of the adversary. The goal is then to identify a set of targets that minimizes traffic quality, i.e., identify a set of edges $\RR \subseteq Q$ in:
$$
\argmin_{Q' \subseteq Q\ :\ |Q'| \leq b} f((V,E \setminus Q'),R,M).
$$
\end{definition}

\begin{definition}[The problem of Minimizing Targets]
\label{def:minimum-traffic-attack}
The problem is defined by a tuple, $(G,Q,R,f,M,\xi)$, where $G=(V,E)$ is a directed network, $Q \subseteq E$ is the set of edges that the adversary can choose from, $R$ is the set of rides, $f$ is an objective function representing traffic quality, $M$ is a model of traffic, and $\xi \in \R$ is the attack efficiency. The goal is then to identify the smallest set of targets that decreases traffic quality below the threshold $\xi$, i.e., identify a set of edges $\RR \subseteq Q$ in:
$$
\argmin_{Q'\subseteq Q\ :\ f((V,E \setminus Q'),R,M)\leq\xi } |Q'|.
$$
\end{definition}

\clearpage

\section{Supplementary note 5: Theoretical results}
\label{sec:note5}

\setcounter{theorem}{0}

\begin{theorem}
The problem of Maximizing Disruption is NP-complete given the objective function $f^*$ and the traffic model $M^*$.
\end{theorem}
\begin{proof}
The decision version of the problem is as follows: given a directed road network $G=(V,E)$, the set of edges that the adversary can choose from $Q \subseteq E$, the set of rides $R$, the objective function $f^*$, the model of traffic $M^*$, the adversary's budget $b \in \N$, and the attack efficiency $\xi \in \R$, does there exist $\RR \subseteq Q$ such that $|\RR| \leq b$, and
$$
f^*((V,E \setminus \RR),R,M^*) \leq \xi.
$$

The problem is trivially in NP, since computing the value of $f^*$ after the removal of a given set of edges $\RR$ can be done in a polynomial time.

We will now prove that the problem is NP-hard. To this end, we will show a reduction from the NP-hard Minimum Multiway Cut problem. This problem is defined by a network $(V,E)$ and a set of terminal nodes $S \subseteq V$. The goal is then to determine whether there exist $c$ edges from $E$ such that after removing these edges from $G$, there exists no path between any two terminal nodes. This problem was proven to be NP-hard for directed unweighted graphs, given the number of terminal nodes $k\geq 2$~\cite{garg1994multiway}.

The main idea of our proof of NP-hardness is as follows. We will first construct an instance of the problem of Maximizing Disruption corresponding to the given instance of the Minimum Multiway Cut problem. We will then show that a solution to the constructed instance of the problem of Maximizing Disruption is also a solution to the given instance of the Minimum Multiway Cut problem. Hence, the existence of a polynomial time algorithm solving the problem of Maximizing Disruption would imply the existence of a polynomial time algorithm solving the NP-hard Minimum Multiway Cut problem, which is impossible unless P=NP.

Given an instance $((V,E),S,c)$ of the Minimum Multiway Cut problem, let us construct the following instance of the problem of Maximizing Disruption:
\begin{itemize}
\item $G=(V,E)$;
\item $Q=E$, \ie, all edges can be chosen by the adversary;
\item $R=\bigcup_{s,s' \in S: s' \neq s} \{(s,s',0),(s',s,0)\}$, \ie, for every pair of different terminal nodes $s,s' \in S$, we create a ride with the starting node $\alpha_i=s$ and the destination node $\beta_i=s'$, and another ride with the starting node $\alpha_i=s'$ and the destination node $\beta_i=s$, where all rides start at midnight, \ie, $\forall_i \theta_i=0$;
\item $f=f^*$;
\item $M=M^*$;
\item $b=c$;
\item $\xi=0$.
\end{itemize}

\noindent Moreover, let the parameters of our model of traffic $M^*$ be the following:

\begin{itemize}
\item The length of every street is exactly the length of the vehicle, which results in every lane of every street having exactly one cell, \ie, $\forall_{e \in E}c_e=1$;
\item Every street has a number of lanes equal to the total number of rides, \ie, $\forall_{e \in E}l_e=|R|$;
\item The maximum speed of each vehicle is $1$, \ie, $v_{max}=1$;
\item The probability of slowing down is 0, \ie, $p_{slow} = 0$.
\end{itemize}

Let us now analyze the time it takes to complete each ride.
Since the number of lanes of every street is equal to the total number of rides, no ride has to wait to enter the street, as there is always at least one empty lane (even if other lanes are occupied by all the remaining rides in $R$). The maximum speed of the vehicle is $v_{max}=1$, and since the probability of randomly slowing down is $p_{slow} = 0$, every ride reaches its maximum speed in the first time step of the model (notice that for every ride the time of start is $\theta_i=0$) and never changes its speed later on. Further, since every street has the length of just one cell, it always takes just one time step to traverse each street on the shortest path from the start to the destination. Therefore, for a given ride $r_i$, the time it takes to reach the destination is simply the distance from the starting node to the destination node expressed as the number of edges, \ie, $\mathcal{T}(r_i,G,M^*) = d_G(\alpha_i,\beta_i)$.
Under these conditions, the objective function is:
$$
f^*(G,R,M^*) = \frac{2}{|S|(|S|-1)} \sum_{s,s' \in S:s \neq s'} \left( \frac{1}{d_G(s,s')} + \frac{1}{d_G(s',s)} \right).
$$

Now, we show that if there exists a solution $\RR$ to the given instance of the Minimum Multiway Cut problem, \ie, a set of $c$ edges such that after the removal of $\RR$ there exists no path between any two terminal nodes, then it is also a solution to the constructed instance of the problem of Maximizing Disruption. Since all the starting and destination nodes in the problem of Maximizing Disruption are terminal nodes from the Minimum Multiway Cut problem, after the removal of $\RR$ there are no paths between them, and the distance between them is $\infty$.
Therefore, after the removal of $\RR$, we have
$$
f^*((V,E\setminus Q^*),R,M^*) = \frac{2}{|S|(|S|-1)} \sum_{s,s' \in S:s \neq s'} \left( \frac{1}{\infty} + \frac{1}{\infty} \right)=0.
$$
Hence, $\RR$ is a solution to the constructed instance of the problem of Maximizing Disruption.

To complete the proof of the theorem, we now show that if there exists a solution $\RR$ to the constructed instance of the problem of Maximizing Disruption, then it is also a solution to the given instance of the Minimum Multiway Cut problem. Since $\RR$ is a solution, after the removal of $\RR$, the value of $f^*$ is zero. If for at least one pair of starting and destination nodes there would exist a path between them, the distance between them would be smaller than $\infty$, and the $\frac{1}{d_G(\alpha_i,\beta_i)}$ component for this pair would cause the value of $f^*$ to be positive. Therefore, since $\RR$ is a solution to the constructed instance of the problem of Maximizing Disruption, there can be no pair of starting and destination nodes with a path between them in $(V,E \setminus \RR)$. However, because of the way we constructed this instance, the pairs of starting and destination nodes are exactly all pairs of terminal nodes from the given instance of the Minimum Multiway Cut problem. Hence, there are no paths between the terminal nodes in $(V,E \setminus \RR)$, and $\RR$ is a solution to the given instance of the Minimum Multiway Cut problem. This concludes the proof.
\end{proof}

\begin{theorem}
The problem of Maximizing Disruption is NP-complete given the objective function $f^*$ and the traffic model $M^\varnothing$.
\end{theorem}
\begin{proof}
The decision version of the problem is as follows: given a directed road network $G=(V,E)$, the set of edges that the adversary can choose from $Q \subseteq E$, the set of rides $R$, the objective function $f^*$, the model of traffic $M^\varnothing$, the adversary's budget $b \in \N$, and the attack efficiency $\xi \in \R$, does there exist $\RR \subseteq Q$ such that $|\RR| \leq b$ and
$$
f^*((V,E \setminus \RR),R,M^\varnothing) \leq \xi.
$$

The problem is trivially in NP, since computing the value of $f^*$ after the removal of a given set of edges $\RR$ can be done in polynomial time.

We will now prove that the problem is NP-hard. To this end, we will show a reduction from the NP-hard Minimum Multiway Cut problem. This problem is defined by a network $(V,E)$ and a set of terminal nodes $S \subseteq V$. The goal is then to determine whether there exist $c$ edges from $E$ such that after removing these edges from $G$, there exists no path between any two terminal nodes. This problem was proven to be NP-hard for directed unweighted graphs, given the number of terminal nodes $k\geq 2$~\cite{garg1994multiway}. The main idea of the proof is the same as for the proof of Theorem~\ref{thrm:npc-our-model}.

Given an instance $((V,E),S,c)$ of the Minimum Multiway Cut problem, let us construct the following instance of the problem of Maximizing Disruption:
\begin{itemize}
\item $G=(V,E)$;
\item $Q=E$, \ie, all edges can be chosen by the adversary;
\item $R=\bigcup_{s,s' \in S: s' \neq s} \{(s,s'),(s',s)\}$, \ie, for every pair of different terminal nodes $s,s' \in S$ we create a ride with starting node $\alpha_i=s$ and destination node $\beta_i=s'$, and another ride with starting node $\alpha_i=s'$ and destination node $\beta_i=s$;
\item $f=f^*$;
\item $M=M^\varnothing$;
\item $b=c$;
\item $\xi=0$.
\end{itemize}

The objective function is then:
$$
f^*(G,R,M^\varnothing) = \frac{2}{|S|(|S|-1)} \sum_{s,s' \in S:s \neq s'} \left( \frac{1}{d_G(s,s')} + \frac{1}{d_G(s',s)} \right).
$$

Notice that this is the same form as the objective function in the proof of Theorem~\ref{thrm:npc-our-model}. Hence the reasoning follows the same as in the proof of Theorem~\ref{thrm:npc-our-model}. We repeat it here for the convenience of the reader.

First, we show that if there exists a solution $\RR$ to the given instance of the Minimum Multiway Cut problem, \ie, a set of $c$ edges such that after the removal of $\RR$ there exists no path between any two terminal nodes, then it is also a solution to the constructed instance of the problem of Maximizing Disruption.
Since all the starting and destination nodes in the problem of Maximizing Disruption are terminal nodes from the Minimum Multiway Cut problem, after the removal of $\RR$, there are no paths between them, and the distance between them is $\infty$.
Therefore, after the removal of $\RR$, we have
$$
f^*((V,E\setminus Q^*),R,M^\varnothing) = \frac{2}{|S|(|S|-1)} \sum_{s,s' \in S:s \neq s'} \left( \frac{1}{\infty} + \frac{1}{\infty} \right)=0.
$$
Hence, $\RR$ is a solution to the constructed instance of the problem of Maximizing Disruption.

To complete the proof of the theorem, we now show that if there exists a solution $\RR$ to the constructed instance of the problem of Maximizing Disruption, then it is also a solution to the given instance of the Minimum Multiway Cut problem. Since $\RR$ is a solution, after the removal of $\RR$, the value of $f^*$ is zero. If for at least one pair of starting and destination nodes there would exist a path between them, the distance between them would be smaller than $\infty$, and the $\frac{1}{d_G(\alpha_i,\beta_i)}$ component for this pair would cause the value of $f^*$ to be positive. Therefore, since $\RR$ is a solution to the constructed instance of the problem of Maximizing Disruption, there can be no pair of starting and destination node with a path between them in $(V,E \setminus \RR)$. However, because of the way we constructed this instance, the pairs of starting and destination nodes are exactly all pairs of terminal nodes from the given instance of the Minimum Multiway Cut problem. Hence, there are no paths between the terminal nodes in $(V,E \setminus \RR)$, and $\RR$ is a solution to the given instance of the Minimum Multiway Cut problem. This concludes the proof.
\end{proof}

\begin{theorem}
The problem of Minimizing Targets is APX-hard given the objective function $f^*$ and either traffic model $M^*$ or $M^\varnothing$. In particular, the problem does not admit a polynomial-time approximation scheme (PTAS) unless P=NP.
\end{theorem}
\begin{proof}
We will prove that the problem is APX-hard by showing an L-reduction from the Minimum Multi-Cut problem. This problem is defined by a network $G=(V,E)$ and a set of source-destination pairs of nodes $S = \left\{ (s_1,t_1), \ldots, (s_{|S|},t_{|S|}) \right\} \subseteq V \times V$. The goal is to find a subset $F \subseteq E$ such that removing edges $F$ from $G$ disconnects every source $s_i$ from its destination $t_i$, and the size of $F$ is minimal. The Minimum Multi-Cut problem was shown to be APX-hard~\cite{dahlhaus1994complexity}.

The L-reduction will have the form of two polynomial time computable functions $p$ and $q$ such that:
\begin{enumerate}[label=(\alph*)]
    \item $p$ constructs an instance of the problem of Minimizing Targets based on an instance of the Minimum Multi-Cut problem;
    \label{pt:apx-hard-1}
    
    \item $q$ constructs a solution to an instance of the Minimum Multi-Cut problem based on a solution of the problem of Minimizing Targets;
    \label{pt:apx-hard-2}
    
    \item For any instance of the Minimum Multi-Cut problem $(G,S)$, an optimal solution $\RR$ to $p(G,S)$ is of the same size as an optimal solution $F^*$ to $(G,S)$, \ie, $|\RR| = |F^*|$;
    \label{pt:apx-hard-3}
    
    \item For every solution $Q'$ to $p(G,S)$, we have that $\left| |F^*| - |q(Q')| \right| = \left| |\RR| - |Q'| \right|$, where $F^*$ is an optimal solution to $(G,S)$ and $\RR$ is an optimal solution to $p(G,S)$.
\label{pt:apx-hard-4}
\end{enumerate}

First, let us define the function $p$. Let $(G,S)$ be an instance of the Minimum Multi-Cut problem. The constructed instance of the problem of Minimizing Targets is then $p(G,S) = (G,Q,R,f,M,\xi)$, where:
\begin{itemize}
    \item The network $G$ is the same as for the Minimum Multi-Cut problem instance;
    \item $Q=E$, \ie, all edges can be chosen by the adversary;
    \item $R=\bigcup_{(s_i,t_i) \in S}\{(s_i,t_i)\}$, \ie, for every pair of source-destination nodes $(s_i,t_i) \in S$, we create a ride with starting node $\alpha_i=s_i$ and destination node $\beta_i=t_i$;
    \item $f=f^*$;
    \item $M$ is the model of traffic, which could be either the simple model $M^\varnothing$ or our model of traffic $M^*$;
    \item $\xi=0$, \ie, the removed set of edges has to decrease the traffic efficiency to 0.
\end{itemize}

\noindent Moreover, in the case of our model of traffic $M^*$, let its parameters be as follows:
\begin{itemize}
    \item The length of every street is exactly the length of the vehicle, which results in every lane of every street having exactly one cell, \ie, $\forall_{e \in E}c_e=1$;
    \item Every street has the number of lanes equal to the total number of rides, \ie, $\forall_{e \in E}l_e=|R|$;
    \item The maximum speed of each vehicle is $1$, \ie, $v_{max}=1$;
    \item The probability of slowing down is 0, \ie, $p_{slow} = 0$.
\end{itemize}
\noindent The simple model of traffic $M^\varnothing$ does not need any additional parameterization.

Note that $p(G,S)$ is a correct definition of an instance of the problem of Minimizing Targets and it is computable in polynomial time (hence point~\ref{pt:apx-hard-1} is fulfilled). Note also that under these assumptions, the value of the objective function is:
$$
f^*(G,R,M) = \frac{1}{|S|} \sum_{(s_i,t_i) \in S} \frac{1}{d_G(s_i,t_i)}.
$$
In case $M=M^\varnothing$, this formula follows directly from the definition of the model (see Definition~\ref{def:simple-model}). In case $M=M^*$, see the proof of Theorem~\ref{thrm:npc-our-model} for an explanation of how the formula for the objective function follows from the set of model parameters chosen above.

Second, we define the function $q$ as simply $q(Q')=Q'$, \ie, the set of edges forming a solution to the Minimum Multi-Cut problem is the same set of edges forming a given solution to the problem of Minimizing Targets.

Now, we will show that if $Q$ is a solution to $p(G,S)$, then $q(Q')$ is indeed a solution to the given instance $(G,S)$ of the Minimum Multi-Cut problem (point~\ref{pt:apx-hard-2}). Note that in order for $Q'$ to be a solution to $p(G,S)$, there cannot be a single pair of source-destination nodes $(s_i,t_i) \in S$ such that there exists a path from $s_i$ to $t_i$ in $(V,E \setminus Q')$. If at least one such pair exists, then we have:
$$
f^*((V,E \setminus Q'),R,M) \geq \frac{1}{|S| d_G(s_i,t_i)} > 0,
$$
whereas a solution $Q'$ has to decrease the traffic efficiency to zero (this is because $\xi = 0$). Hence, if there are no paths from the source to destination in $(V,E \setminus Q')$, then $q(Q')=Q'$ is also a solution to the given instance $(G,S)$ of the Minimum Multi-Cut problem.

Now, we will show that for any instance of the Minimum Multi-Cut problem, $(G,S)$, an optimal solution $\RR$ to $p(G,S)$ is of the same size as an optimal solution $F^*$ to $(G,S)$ (point~\ref{pt:apx-hard-3}). We will show this by contradiction. Assume there exists an optimal solution $\RR$ to $p(G,S)$ smaller than the optimal solution $F^*$ to $(G,S)$. However, as we have shown in the paragraph above, the removal of $\RR$ from $G$ disconnects all pairs in $S$, hence $\RR$ is also a solution to $(G,S)$. As $\RR$ is smaller than $F^*$, $F^*$ cannot be optimal. Similarly, assume there exists an optimal solution $F^*$ to $(G,S)$ smaller than an optimal solution $\RR$ to $p(G,S)$. However, as the removal of $F^*$ from $G$ disconnects all pairs in $S$, the value of the objective function is $f^*((V, E \setminus F^*),R,M)=0$ and $F^*$ is also a solution to $p(G,S)$. Since $F^*$ is smaller than $\RR$, it follows that $\RR$ cannot be an optimal solution to $p(G,S)$.

To complete the proof of the theorem, we will now show that for every solution $Q'$ to $p(G,S)$ we have $\left| |F^*| - |q(Q')| \right| = \left| |\RR| - |Q'| \right|$, where $F^*$ is an optimal solution to $(G,S)$ and $\RR$ is an optimal solution to $p(G,S)$ (point~\ref{pt:apx-hard-4}). In the paragraph above, we showed that the optimal solutions to $(G,S)$ and $p(G,S)$ are of the same size, \ie, $|F^*|=|\RR|$. Moreover, notice that $q(Q')=Q'$. Hence, we have $\left| |F^*| - |q(Q')| \right| = \left| |\RR| - |Q'| \right|$.

Therefore, we showed that functions $p$ and $q$ define an L-reduction from the Minimum Multi-Cut problem to the problem of Minimizing Targets, given the objective function $f^*$ and either the simple model of traffic $M^\varnothing$ or our model of traffic $M^*$. Since the Minimum Multi-Cut problem is APX-hard~\cite{dahlhaus1994complexity}, the problem of Minimizing Targets is also APX-hard, as L-reduction preserves APX-hardness.

Finally, it is well-known that if there exists a PTAS for some APX-hard problem, then P=NP~\cite{arora1998proof}. Hence, the problem of Minimizing Targets does not admit a PTAS, unless P=NP. 
\end{proof}

\end{document}